\documentclass[submission,copyright]{eptcs}
\providecommand{\event}{Linearity and TLLA 2018} 
\usepackage{underscore}           

\usepackage{amsmath}
\usepackage{amssymb}
\usepackage{stmaryrd}
\usepackage{proof}
\usepackage{graphicx}
\usepackage{listings}
\usepackage{enumitem}
\usepackage{hyperref}
\usepackage[all,cmtip]{xy}

\usepackage[dvipsnames]{xcolor}

\newcommand{\bccolor}{Green}
\newcommand\cte[1]{\text{\ttfamily\upshape #1}}  
\newcommand\ord[5]{\cte{order}(#1,~#2,~#3,~#4,~#5)}
\newcommand\ordN[6]{\cte{ordIn}(#1,~#2,~#3,~#4,~#5,~#6)}
\newcommand\ordQ[1]{\cte{queue}(#1)}
\newcommand\prcQ[3]{\cte{priceQ}(#1,~#2,~#3)}
\newcommand\actP[2]{\cte{actPrices}(#1,~#2)}
\newcommand\tm[1]{\cte{time}(#1)}
\newcommand\natless[2]{\cte{{\color{\bccolor} nat-less}}(#1,~#2)}
\newcommand\natgreat[2]{\cte{{\color{\bccolor} nat-great}}(#1,~#2)}
\newcommand\nateq[2]{\cte{{\color{\bccolor} nat-equal}}(#1,~#2)}
\newcommand\natminus[2]{\cte{{\color{\bccolor} nat-minus}}(#1,~#2)}
\newcommand\dual[2]{\cte{{\color{\bccolor} dual}}(#1,~#2)}
\newcommand\maxP[2]{\cte{{\color{\bccolor} maxP}}(#1,~#2)}
\newcommand\minP[2]{\cte{{\color{\bccolor} minP}}(#1,~#2)}
\newcommand\store[3]{\cte{{\color{\bccolor} store}}(#1,~#2,~#3)}
\newcommand\exchange[4]{{\color{\bccolor} \cte{exchange}}(#1,~#2,~#3,~#4)}
\newcommand\mktex[3]{{\color{\bccolor} \cte{mktExchange}}(#1,~#2,~#3)}
\newcommand\ins[3]{{\color{\bccolor} \cte{insert}}(#1,~#2,~#3)}
\newcommand\remove[3]{{\color{\bccolor} \cte{remove}}(#1,~#2,~#3)}

\newcommand\removeF[3]{{\color{\bccolor} \cte{removeF}}(#1,~#2,~#3)}
\newcommand\notInList[2]{\cte{{\color{\bccolor} notInList}}(#1,~#2)}
\newcommand\notInListF[2]{\cte{{\color{\bccolor} notInListF}}(#1,~#2)}

\newcommand\inListF[2]{\cte{{\color{\bccolor} inListF}}(#1,~#2)}
\newcommand\consP[4]{\cte{consP}(#1,~#2,~#3,~#4)}
\newcommand\ext[5]{\cte{{\color{\bccolor} extendP}}(#1,~#2,~#3,~#4,~#5)}
\newcommand\enq[3]{\cte{{\color{\bccolor} enq}}(#1,~#2,~#3)}

\newcommand\cbuy{\cte{buy}}
\newcommand\csell{\cte{sell}}
\newcommand\cnil{\cte{nil}}
\newcommand\cnilp{\cte{nilP}}
\newcommand\cniln{\cte{nilN}}

\newcommand\cgenp[4]{\cte{gen}(#1,~#2,~#3,~#4)}

\newcommand\signcl{\Sigma_{NLC}}

\newcommand\cbid{\mathit{bid}}
\newcommand\cask{\mathit{ask}}

\newtheorem{property}{Property}
\newtheorem{thm}[property]{Theorem}
\newtheorem{defn}[property]{Definition}
\newenvironment{proof}[1][Proof]{\begin{trivlist}
\item[\hskip \labelsep {\bfseries #1}]}{\end{trivlist}}

\title{Formalization of Automated Trading Systems\\
  in a Concurrent Linear Framework\thanks{This paper was made possible by grant
  NPRP 7-988-1-178 from the Qatar National Research Fund (a member of the Qatar
  Foundation). The statements made herein are solely the responsibility of the
  authors.}
}
\author{
Iliano Cervesato \qquad\qquad Sharjeel Khan \qquad\qquad Giselle Reis \qquad\qquad Dragi\v{s}a \v{Z}uni\'{c}
\institute{Carnegie Mellon University}
\email{iliano@cmu.edu \quad\quad smkhan@andrew.cmu.edu \quad\qquad giselle@cmu.edu \quad\qquad dzunic@andrew.cmu.edu}
}
\def\titlerunning{Formalization of Automated Trading Systems}
\def\authorrunning{I. Cervesato, S. Khan, G. Reis \& D. \v{Z}uni\'{c}}

\begin{document}
\maketitle

\begin{abstract}
  We present a declarative and modular specification of an automated
  trading system (ATS) in the concurrent linear framework CLF\@.  We
  implemented it in Celf, a CLF type checker which also supports
  executing CLF specifications.  We outline the verification of two
  representative properties of trading systems using generative
  grammars, an approach to reasoning about CLF specifications.
\end{abstract}

\section{Introduction}

Trading systems are platforms where buy and sell orders are
automatically matched. Matchings are executed according to the
operational specification of the system. In order to guarantee trading
fairness, these systems must meet the requirements of regulatory
bodies, in addition to any internal requirement of the trading
institution. However, both specifications and requirements are
presented in natural language which leaves space for ambiguity and
interpretation errors.  As a result, it is difficult to guarantee
regulatory compliance~\cite{DeBel93}.  For example, the main US
regulator, the Securities and Exchanges Commission (SEC), has fined
several companies, including Deutsche Bank (37M in 2016),
Barclay's Capital (70M in 2016), Credit Suisse (84M in 2016), UBS
(19.5M in 2015) and many others~\cite{Freedman2015} for non-compliance.

Modern trading systems are complex pieces of software with intricate
and sensitive rules of operation.  Moreover they are in a state of
continuous change as they strive to support new client requirements
and new order types.  Therefore it is difficult to attest that they
satisfy all requirements at all times using standard software testing
approaches.  Even as regulatory bodies recently demand that systems
must be ``fully tested''~\cite{FCA2018}, experience has shown that
(possibly unintentional) violations often originate from unforeseen
interactions between order types~\cite{PassmoreIgnatovich2017}.

Formalization and formal reasoning can play a big role in mitigating
these problems. They provide methods to verify properties of complex
and infinite state space systems with certainty, and have already been
applied in fields ranging from microprocessor design~\cite{Jones2001},
avionics~\cite{Souyris2009}, election security~\cite{Pattison15AJCAI},
and financial derivative
contracts~\cite{PeytonJones2000,BahrElsman2015}.  Trading systems are
a prime candidate~as~well.

In this paper we use the logical framework CLF~\cite{clftechrep1} to
specify and reason about trading systems. CLF is a linear concurrent
extension of the long-established LF
framework~\cite{Harper93}. Linearity enables natural encoding of state
transition, where facts are consumed and produced thereby changing the
system's state. The concurrent nature of CLF is convenient to account
for the possible orderings of exchanges.

The contributions of this research are twofold: (1) We formally define
an archetypal automated trading system in CLF~\cite{clftechrep1} and
implement it as an executable specification in Celf.  (2) We
demonstrate how to prove some properties about the specification using
generative grammars~\cite{Simmons12}, a technique for meta-reasoning
in CLF\@.

The paper is organized as follows: Section~\ref{sec:CLF} introduces the
concurrent logical framework CLF\@. Section~\ref{sec:ATS} introduces the core
concepts related to automated trading systems (ATS), followed by
Section~\ref{sec:formlimit}, which presents the formalization of an ATS in
CLF/Celf. Section~\ref{sec:proof} contains proofs of two properties based on
generative grammars, going towards automated reasoning in CLF\@.  We conclude and
outline possible further developments in Section~\ref{sec:conclusion}.

\section{Concurrent Linear Logic and Celf}
\label{sec:CLF}

The logical framework CLF~\cite{clftechrep1} is based on a fragment of
intuitionistic linear logic. It extends the logical framework
LF~\cite{Harper93} with the linear connectives $\multimap$,
$\binampersand$, $\top$, $\otimes$, $1$ and $!$ to obtain a
resource-aware framework with a satisfactory representation of
concurrency. The rules of the system impose a discipline on when the
synchronous connectives $\otimes$, $1$ and $!$ are decomposed, thus
still retaining enough determinism to allow for its use as a logical
framework.  Being a type-theoretical framework, CLF unifies
implication and universal quantification as the dependent product
construct.  For simplicity we present only the logical fragment of CLF
(i.e., without terms) needed for our encodings. A detailed description
of the full framework can be found in~\cite{clftechrep1}.

We divide the formulas in this fragment of CLF into two classes:
\emph{negative} and \emph{positive}.  Negative formulas have right
invertible rules and positive formulas have left invertible rules.
Their grammar is:
{\small
\[
\begin{array}{rllr}
N, M & ::= & a \multimap N \mid a \rightarrow N \mid \{ P \} \mid \forall x. N \mid a & \text{(negative formulas)}\\
P, Q & ::= & P \otimes Q \mid 1 \mid ! a \mid a & \text{(positive formulas)}
\end{array}
\]}%
where $a$ is an atom (i.e., a predicate).  Positive formulas are
enclosed in the lax modality $\{\cdot\}$, which ensures that their
decomposition happens atomically.

\newcommand{\negR}[3]{#1; #2 \vdash #3}
\newcommand{\negL}[3]{#1; #2 \vdash #3}
\newcommand{\posR}[3]{#1; #2 \vdash {\color{red} #3}}
\newcommand{\posL}[4]{#1; #2; {\color{blue} #3} \vdash #4}

\begin{figure}
\small
\[
\infer[1_l]
  {\posL{\Gamma}{\Delta}{\Psi, 1}{P_0}}
  {\posL{\Gamma}{\Delta}{\Psi}{P_0}}
\quad
\infer[\otimes_l]
  {\posL{\Gamma}{\Delta}{\Psi, P \otimes Q}{P_0}}
  {\posL{\Gamma}{\Delta}{\Psi, P, Q}{P_0}}
\quad
\infer[!_l]
  {\posL{\Gamma}{\Delta}{\Psi, !a}{P_0}}
  {\posL{\Gamma, a}{\Delta}{\Psi}{P_0}}
\quad
\infer[\mathsf{st}]
  {\posL{\Gamma}{\Delta}{\Psi, a}{P_0}}
  {\posL{\Gamma}{\Delta, a}{\Psi}{P_0}}
\quad
\infer[\mathsf{L}]
  {\posL{\Gamma}{\Delta}{\cdot}{P_0}}
  {\negL{\Gamma}{\Delta}{P_0}}
\]
\\[-3.5ex]
\[
\infer[1_r]
  {\posR{\Gamma}{\cdot}{1}}
  {}
\quad
\infer[\otimes_r]
  {\posR{\Gamma}{\Delta_1, \Delta_2}{P \otimes Q}}
  {\posR{\Gamma}{\Delta_1}{P}
  &\posR{\Gamma}{\Delta_2}{Q}}
\quad
\infer[!_r]
  {\posR{\Gamma}{\cdot}{!a}}
  {\posR{\Gamma}{\cdot}{a}}
\quad
\infer[\mathsf{R}]
  {\posR{\Gamma}{\Delta}{a}}
  {\negR{\Gamma}{\Delta}{a}}
\quad
\infer[\mathsf{init}]
  {\negL{\Gamma}{a}{a}}
  {}
\]
\\[-3.5ex]
\[
\infer[\multimap_l]
  {\negL{\Gamma}{\Delta_1, \Delta_2, a \multimap N}{F}}
  {\posR{\Gamma}{\Delta_1}{a}
  &\negL{\Gamma}{\Delta_2, N}{F}}
\quad
\infer[\rightarrow_l]
  {\negL{\Gamma}{\Delta, a \rightarrow N}{F}}
  {\posR{\Gamma}{\cdot}{a}
  &\negL{\Gamma}{\Delta, N}{F}}
\quad
\infer[\forall_l]
  {\negL{\Gamma}{\Delta, \forall x. N}{F}}
  {\negL{\Gamma}{\Delta, N[x \mapsto t]}{F}}
\quad
\infer[\{\}_l]
  {\negL{\Gamma}{\Delta, \{ P \}}{P_0}}
  {\posL{\Gamma}{\Delta}{P}{P_0}}
\]
\\[-3.5ex]
\[
\infer[\multimap_r]
  {\negR{\Gamma}{\Delta}{a \multimap N}}
  {\negR{\Gamma}{\Delta, a}{N}}
\quad
\infer[\rightarrow_r]
  {\negR{\Gamma}{\Delta}{a \rightarrow N}}
  {\negR{\Gamma, a}{\Delta}{N}}
\quad
\infer[\forall_r]
  {\negR{\Gamma}{\Delta}{\forall x. N}}
  {\negR{\Gamma}{\Delta}{N[x \mapsto \alpha]}}
\quad
\infer[\{\}_r]
  {\negR{\Gamma}{\Delta}{\{ P \}}}
  {\posR{\Gamma}{\Delta}{P}}
\quad
\infer[\mathsf{cont}]
  {\negL{\Gamma, N}{\Delta}{C}}
  {\negL{\Gamma, N}{\Delta, N}{C}}
\]

\caption{Sequent calculus for a fragment of CLF\@. $N$ is a negative formula,
$P$ and $Q$ are positive formulas, $P_0$ is either an atom or $\{ P \}$, $F$ is any
formula, $a$ is an atom, $\alpha$ is an eigenvariable and $t$ is a term.}
\label{fig:clf-seq}
\end{figure}

The sequent calculus proof system for this fragment of CLF is
presented in Figure~\ref{fig:clf-seq}.  The sequents make use of
either two or three contexts on the left: $\Gamma$ contains
unrestricted formulas, $\Delta$ contains linear formulas and $\Psi$,
when present, contains positive formulas. The decomposition phase of a
positive formula is indicated in {\color{red}red} on the right and in
{\color{blue}blue} on the left.
These phases end (by means of rules \textsf{L} or \textsf{R}) after the formula
is completely decomposed.

Since CLF has both the linear and intuitionistic implications, we can
specify computation in two different ways. Simplifying somewhat, linear implication formulas
are interpreted as multiset rewriting: the bounded resources on the
left are consumed and those on the right are produced. State
transitions can be modeled naturally this way.  Intuitionistic
implication formulas are interpreted as backward-chaining rules
\emph{\`a la} Prolog, providing a way to compute solutions for a
predicate by matching it with the head (rightmost predicate) of a
rule and solving the body. In this paper, predicates defined by backward-chaining
rules are written in {\color{\bccolor}green}.

The majority of our encoding involves rules in the following shape (for atomic
$p_i$ and $q_i$):\\
$p_1 \otimes ... \otimes p_n \multimap \{ q_1 \otimes ... \otimes q_m \}$
which is the uncurried version of:
$p_1 \multimap ... \multimap p_n \multimap \{ q_1 \otimes ... \otimes q_m \}$.

This framework is implemented as the tool Celf
(\url{https://clf.github.io/celf/}) which we used for the encodings.
Following the tool's convention, variable names start with an upper-case letter.

\begin{figure*}
\begin{minipage}[c]{1.\textwidth}
\begin{center}
\includegraphics[width=.75\columnwidth]{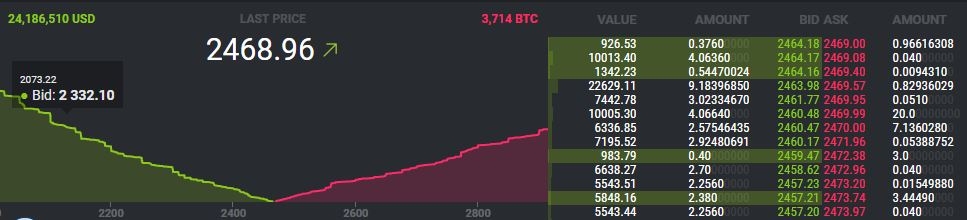}
\end{center}
\end{minipage}
\caption{Visualization of the market view}
\label{fig:marketView}
\end{figure*}

\section{Automated Trading System (ATS)}
\label{sec:ATS}

Real life trading systems differ in the details of how they manage orders (there
are hundreds of order types in use~\cite{Mackintosh2014}). However, there is a
certain common core that guides all those trading systems, and which embodies
the market logic of trading on an exchange.  We have formalized those elements
in what we call an automated trading system, or an ATS\@. In what follows we
introduce the basic notions.

An \emph{order} is an investor's instruction to a broker to buy or sell
securities (or any asset type which can be traded).  They enter an ATS
sequentially and are \emph{exchanged} when successfully matched against opposite
order(s).  In this paper, we will be concerned with \emph{limit}, \emph{market}
and \emph{cancel} orders.

A \emph{limit order} has a specified limit price, meaning that it will trade at
that price or better.  In the case of a limit order to sell, a limit price $P$
means that the security will be sold at the best available price in the market,
but no less than $P$.  And dually for buy orders. If no exchange is possible,
the order stays in the market waiting to be exchanged -- these are called
\emph{resident orders}.
A \emph{market order} does not specify the price, and will be immediately
matched against opposite orders in the market. If none are available, the order
is discarded. A \emph{cancel order} is an instruction to remove a resident order
from the market.

A \emph{matching algorithm} determines how resident orders are
prioritized for exchange, essentially defining the mode of operation
of a given ATS\@.  The most common one is \emph{price/time priority}.
Resident orders are first ranked according to their price
(increasingly for sell and decreasingly for buy orders); orders with
the same price are then ranked depending on when they entered.

Figure~\ref{fig:marketView} presents a visualization of a (Bitcoin)
market.  The left-hand side (green) contains resident buy orders,
while the right-hand side (pink) contains resident sell orders. The price
offered by the most expensive buy order is called \emph{bid} and the
cheapest sell order is called \emph{ask}.  The point where they
(almost) meet is the \emph{bid-ask} spread, which, at that particular
moment, was around 2,468 USD\@.

Standard regulatory requirements for real world trading
systems include: the bid price is always strictly less than the ask price (i.e., no
locked -- $\cbid$ is equal to $\cask$ -- or crossed -- $\cbid$ is greater than $\cask$ --
states), the trade always takes place at either $\cbid$ or $\cask$, the
price/time priority is always respected when exchanging orders, the order priority
function is transitive, among others.

\section{Formalization of an ATS}
\label{sec:formlimit}

We have formalized the most common components of an ATS in the logical
framework CLF and implemented them in Celf. The formalization of an ATS is
divided into three parts.  First, we represent the market infrastructure using
some auxiliary data structures.  Then we determine how to represent the basic
order types and how they are organized for processing. Finally we encode the
exchange rules which act on incoming orders.

Since we are using a linear framework, the state of the system is
naturally represented by a set of facts which hold at that point in
time.  Each rule consumes some of these facts and generates others,
thus reaching a new state.  Many operations are dual for buy and sell
orders, so, whenever possible, predicates and rules are parameterized
by the action (\cte{sell} or \cte{buy}, generically denoted $A$). The
machinery needed in our formalization includes natural numbers, lists
and queues. Their encoding relies on the backward-chaining semantics
of Celf.

The full encoding can be found at
\url{https://github.com/Sharjeel-Khan/financialCLF}.

\subsection{Infrastructure}

The trading system's infrastructure is represented by the following four linear
predicates:
{\small
$$
\ordQ{Q} \qquad \prcQ{A}{P}{Q} \qquad \actP{A}{L} \qquad \tm{T}
$$}%
Predicate $\ordQ{Q}$ represents the queue in which orders are inserted for
processing. As orders arrive in the market, they are assigned a
timestamp and added to $Q$.  For an action $A$ and price $P$, the
queue $Q$ in $\prcQ{A}{P}{Q}$ contains all resident orders with those
attributes. Due to how orders are processed, the queue is sorted in
ascending order of timestamp.  We maintain the invariant that price
queues are never empty.  Price queues correspond to columns in the
graph of Figure~\ref{fig:marketView}.  For an action $A$, the list $L$
in $\actP{A}{L}$ contains the exchange prices available in the market,
i.e., all the prices on the $x$-axis of Figure~\ref{fig:marketView}
with non-empty columns.  Note that the bid price is the maximum of $L$
when $A$ is \cte{buy} and the ask price is the minimum when $A$ is \cte{sell}.  The
time is represented by the fact $\tm{T}$ and increases as the state
changes.

The \cte{begin} fact is the entry point in our formalization. This
fact starts the ATS\@.  It is rewritten to an empty order queue, empty
active price lists for \cte{buy} and \cte{sell}, and the zero time:

{\small
\[
\cte{begin} \multimap \{ \ordQ{\cte{empty}} \otimes \actP{\cte{buy}}{\cnil} \otimes \actP{\cte{sell}}{\cnil} \otimes \tm{\cte{z}} \}
\]
}

\subsection{Orders' Structure}

An order is represented by a linear fact $\ord{O}{A}{P}{ID}{N}$, where $O$ is
the type of order, $A$ is an action, $P$ is the order price, $\mathit{ID}$ is
the identifier of the order and $N$ is the quantity. $P$, $\mathit{ID}$ and $N$
are natural numbers. In this paper, $O$ is one of \cte{limit}, \cte{market}, or
\cte{cancel}. An order predicate in the context is consumed and added to the
order queue for processing via the following rule:

{\small
\[
\begin{array}{rl}
&\ord{O}{A}{P}{ID}{N} \otimes \ordQ{Q} \otimes \tm{T} \otimes \enq{Q}{\ordN{O}{A}{P}{ID}{N}{T}}{Q'} \\
\multimap &\{ \ordQ{Q'} \otimes \tm{\cte{s}(T)} \}
\end{array}
\]
}

The predicate is transformed into a term $\ordN{O}{A}{P}{ID}{N}{T}$
containing the same arguments plus the timestamp $T$. This term is
added to the order queue $Q$ by the (backward-chaining) predicate
{\color{\bccolor}\cte{enq}}.  This queue allows the sequential processing of
orders given their time of arrival in the market, thus simulating what happens
in reality. The timestamp is also used to define resident order priority.
Sequentiality is guaranteed as all state transition rules act only on the first
order in the queue. Nevertheless, due to Celf's non-determinism, orders are
added to the queue in an arbitrary order.

\subsection{Limit Orders}

According to the matching logic, there are two basic actions for every limit
order in the queue: exchange (partially or completely) or add to the market
(becomes resident).  The action taken depends on the order's limit price (at
which it is willing to trade), the bid and ask prices, as well as the quantity
of resident orders\footnote{Sometimes an incoming limit order will be partially
filled, with the remainder (once resident orders that match this limit price are
filled) becoming a new resident order.}.

\paragraph{Adding orders to the market}
An order is added to the market when its limit price $P$ is such that it cannot
be exchanged against opposite resident orders. Namely when $P < \cask$ in the
case of a buy order, and when $P > \cbid$ in the case of a sell order.  There
are two rules for adding an order, depending on whether there are resident
orders at that price in the market or not (see Figure~\ref{fig:limitOrdersAdd}).
The (backward chaining) predicate {\color{\bccolor}\cte{store}} is provable when
the order cannot be exchanged.

\begin{figure*}[h]
{\footnotesize
\[
\begin{array}{rl}
\cte{limit/empty:} & \ordQ{\cte{front}(\ordN{\cte{limit}}{A}{P}{ID}{N}{T},Q)} \otimes \dual{A}{A'} \otimes \actP{A'}{L'} \otimes \phantom{.} \\
                   & \store{A}{L'}{P} \otimes \actP{A}{L} \otimes \notInList{L}{P} \otimes \ins{L}{P}{LP} \otimes \tm{T} \\
         \multimap & \{ \ordQ{Q} \otimes \actP{A'}{L'} \otimes \prcQ{A}{P}{\consP{ID}{N}{T}{\cte{nilP}}}  \\
                   & \phantom{\{} \otimes \actP{A}{LP} \otimes \tm{\cte{s}(T)} \} \\
    \\
\cte{limit/queue:} & \ordQ{\cte{front}(\ordN{\cte{limit}}{A}{P}{ID}{N}{T},Q)} \otimes \dual{A}{A'} \otimes \actP{A'}{L'} \otimes \phantom{.} \\
                   & \store{A}{L'}{P} \otimes \prcQ{A}{P}{PQ} \otimes \ext{PQ}{ID}{N}{T}{PQ'} \otimes \tm{T} \\
    \multimap &\{ \ordQ{Q} \otimes \actP{A'}{L'} \otimes \prcQ{A}{P}{PQ'} \otimes \tm{\cte{s}(T)} \}
\end{array}
\]
}
\caption{Adding limit orders to the market}
\label{fig:limitOrdersAdd}
\end{figure*}

The first line is the same for both. Given the order at the \cte{front} of the
order queue, the predicate \cte{dual} will bind $A'$ to the dual action of $A$
(i.e., if $A$ is \cte{buy}, $A'$ will be \cte{sell} and vice-versa).  Then
\actP{$A'$}{$L'$} binds $L'$ to the list of active prices of $A'$.  The incoming
order can be added to the market only if there is no dual resident order at an
acceptable price. For example, if $A$ is \cte{buy} at price $P$, any resident
\cte{sell} order with price $P$ \emph{or less} would be an acceptable match.
The predicate \cte{store} holds iff there is no acceptable match.

The second line of each rule distinguishes whether the new order to be added is
the first one at that price (\cte{limit/empty} rule -- $\notInList{L}{P}$) or
not.  If it is, the active price list is updated by backward chaining on
$\ins{L}{P}{LP}$, and rewriting $\actP{A}{L}$ to $\actP{A}{LP}$.  Additionally,
a new price queue is created with that order alone:
$\prcQ{A}{P}{\consP{ID}{N}{T}{\cte{nilP}}}$.  If there are resident orders at
the same price (and action), the existing price queue is extended with the new
order by backward chaining on $\ext{PQ}{ID}{N}{T}{PQ'}$ and by rewriting
~$\prcQ{A}{P}{PQ}$ to $\prcQ{A}{P}{PQ'}$. Both rules increment the time by one
unit.

\paragraph{Exchanging orders}
The rules for exchanging orders are presented in
Figure~\ref{fig:limitOrdersEx}. A limit order is exchanged when its limit
price $P$ satisfies $P \leq \cbid$, in the case of sell orders, or $P \geq
\cask$ for buy orders.

\begin{figure*}[ht]
{\footnotesize
\[
\begin{array}{rl}
\cte{limit/1:} & \ordQ{\cte{front}(\ordN{\cte{limit}}{A}{P}{ID}{N}{T},Q)} \otimes \dual{A}{A'} \otimes \actP{A'}{L'} \otimes \phantom{a} \\
               & \exchange{A}{L'}{P}{X} \otimes \prcQ{A'}{X}{\consP{ID'}{N'}{T'}{\cte{nilP}}} \otimes \remove{L'}{X}{L''}~ \otimes\\
               & \nateq{N}{N'} \otimes \tm{T} \\
     \multimap & \{ \ordQ{Q} \otimes \actP{A'}{L''} \otimes \tm{\cte{s}(T)} \} \\
\\
\cte{limit/2:} & \ordQ{\cte{front}(\ordN{\cte{limit}}{A}{P}{ID}{N}{T},Q)} \otimes \dual{A}{A'} \otimes \actP{A'}{L'} \otimes \phantom{a} \\
               & \exchange{A}{L'}{P}{X} \otimes \prcQ{A'}{X}{\consP{ID'}{N'}{T'}{\consP{ID1}{N1}{T1}{L}}} ~\otimes\\
               & \nateq{N}{N'} \otimes \tm{T} \\
     \multimap & \{ \ordQ{Q} \otimes \actP{A'}{L'} \otimes \prcQ{A'}{X}{\consP{ID1}{N1}{T1}{L}} \otimes \tm{\cte{s}(T)} \} \\
\\
\cte{limit/3:} & \ordQ{\cte{front}(\ordN{\cte{limit}}{A}{P}{ID}{N}{T},Q)} \otimes \dual{A}{A'} \otimes \actP{A'}{L'} \otimes \phantom{a} \\
               & \exchange{A}{L'}{P}{X} \otimes \prcQ{A'}{X}{\consP{ID'}{N'}{T'}{\cte{nilP}}} \otimes \remove{L'}{X}{L''} \otimes \phantom{a} \\
               & \natgreat{N}{N'} \otimes \natminus{N,N'}{N''} \\
     \multimap & \{ \ordQ{\cte{front}(\ordN{\cte{limit}}{A}{P}{ID}{N''}{T},Q)} \otimes \actP{A'}{L''} \} \\
\\
\cte{limit/4:} & \ordQ{\cte{front}(\ordN{\cte{limit}}{A}{P}{ID}{N}{T},Q)} \otimes \dual{A}{A'} \otimes \actP{A'}{L'} \otimes \phantom{a} \\
               & \exchange{A}{L'}{P}{X} \otimes \prcQ{A'}{X}{\consP{ID'}{N'}{T'}{\consP{ID1}{N1}{T1}{L}}} \otimes \phantom{a} \\
               & \natgreat{N}{N'} \otimes \natminus{N,N'}{N''} \\
     \multimap & \{ \ordQ{\cte{front}(\ordN{\cte{limit}}{A}{P}{ID}{N''}{T},Q)} \otimes \actP{A'}{L'} \otimes \phantom{a} \\
               & \prcQ{A'}{X}{\consP{ID1}{N1}{T1}{L}} \} \\
\\
\cte{limit/5:} & \ordQ{\cte{front}(\ordN{\cte{limit}}{A}{P}{ID}{N}{T},Q)} \otimes \dual{A}{A'} \otimes \actP{A'}{L'} \otimes \phantom{a} \\
               & \exchange{A}{L'}{P}{X} \otimes \prcQ{A'}{X}{\consP{ID'}{N'}{T'}{L}} ~\otimes\\
               & \natless{N}{N'} \otimes \natminus{N',N}{N''} \otimes \tm{T} \\
     \multimap & \{ \ordQ{Q} \otimes \actP{A'}{L'} \otimes \prcQ{A'}{X}{\consP{ID'}{N''}{T'}{L}} \otimes \tm{\cte{s}(T)} \}
\end{array}
\]
}
\caption{Exchanging limit orders}
\label{fig:limitOrdersEx}
\end{figure*}

The (backward chaining) predicate {\color{\bccolor}\cte{exchange}} binds $X$ to
the exchange price (either $\cbid$ or $\cask$). We distinguish between an
incoming order that ``consumes'' all the quantity available at price $X$, or
only a part of the combined quantity available. The arithmetic comparison and
operations are implemented in the usual backward-chaining way using a unary
representation of natural numbers.

All five rules start the same way: the first line binds $L'$ to the list of
active prices of the dual orders.  The backward-chaining predicate
$\exchange{A}{L'}{P}{X}$ holds iff there is a matching resident order. In this
case, it binds $X$ to the best available market price.
The first order in the price queue for $X$ has priority and will be exchanged.
Let $N$ be the quantity in the incoming order and $N'$ be the quantity of the
resident order with highest priority. There are three cases:

\begin{itemize}[leftmargin=*]
\item $N = N'$ (rules \cte{limit/1} and \cte{limit/2}):
Both orders will be completely exchanged. The incoming order is removed from the
order queue by continuing with $\ordQ{Q}$. The resident order is removed from
its price queue and we distinguish two cases:
  \begin{itemize}
  \item It is the last element in the queue (\cte{limit/1}): then the fact
  $\prcQ{\_}{\_}{\_}$ is not rewritten and $X$ is removed from the list of
  active prices by the backward chaining predicate $\remove{L'}{X}{L''}$. This
  list is rewritten from $\actP{A'}{L'}$ to $\actP{A'}{L''}$.
  \item Otherwise (\cte{limit/2}), the resident order is removed from the
  price queue by rewriting\newline
  $\prcQ{A'}{X}{\consP{ID'}{N}{T'}{\consP{ID1}{N1}{T1}{L}}}$ to\newline
  $\prcQ{A'}{X}{\consP{ID1}{N1}{T1}{L}}$.
  \end{itemize}
\item $N > N'$ (rules \cte{limit/3} and \cte{limit/4}): The incoming order will be
partially exchanged, and not leave the order queue as long as there are matching
resident orders. At each exchange its quantity is updated to $N''$, the
difference between $N$ and $N'$, computed by the backward-chaining predicate
$\natminus{N,N'}{N''}$. The order queue is rewritten to\\
$\ordQ{\cte{front}(\ordN{\cte{limit}}{A}{P}{ID}{N''}{T},Q)}$.
The resident order will be completely consumed and removed from the market.
Therefore, we need to distinguish two cases as before (last element in its
price queue -- \cte{limit/3} -- or not -- \cte{limit/4}).
\item $N < N'$ (rule \cte{limit/5}): The incoming order is completely exchanged and
removed from the order queue (rewritten to $\ordQ{Q}$). The resident order is
partially exchanged and its quantity is updated to $N''$, computed by the
backward chaining rule $\natminus{N',N}{N''}$.
Notice that the price queue is rewritten with the order in the same position, so
its priority does not change.
\end{itemize}

By convention, the time is only updated once an order is completely processed
and removed from the order queue.

\subsection{Market Orders}
\label{subs:market}

A market order is meant to be exchanged immediately at current market prices. As
long as there are available sellers or buyers, market orders are exchanged. The
remaining part of a market order is discarded. Certainty of execution is a
priority over the price of execution.

In this case, there are no rules for adding/storing of market orders. The
exchange rules are similar to the ones for limit orders, with subtle
differences.
Market orders do not have a desired price as an attribute, they only have a
desired quantity. Therefore exchanging is continued as long as the quantity was not
reached and there are available resident orders (the price $P$ is nominally
presented but it is never used).
In the rules, the predicate $\exchange{A}{L}{P}{X}$ is replaced with
$\mktex{A}{L}{X}$, which simply binds $X$ to the best offer in the
market.  In the unlikely event that there are no more dual resident orders
(verified by checking if $L$ in $\actP{A'}{L}$ is empty), the order is removed
from the order queue.

\begin{figure*}[ht]
{\footnotesize
\[
\begin{array}{rl}
\cte{market/empty:}
        & \ordQ{\cte{front}(\ordN{\cte{market}}{A}{P}{ID}{N}{T},Q)} \otimes \dual{A}{A'} \otimes \actP{A'}{\cniln} ~\otimes \\
        & \tm{T} \\
     \multimap & \{ \ordQ{Q} \otimes \actP{A'}{\cniln} \otimes \tm{\cte{s}(T)} \} \\[2mm]
\cte{market/1:} & \ordQ{\cte{front}(\ordN{\cte{market}}{A}{P}{ID}{N}{T},Q)} \otimes \dual{A}{A'} \otimes \actP{A'}{L'} ~\otimes\\
                & \mktex{A}{L'}{Y} \otimes \prcQ{A'}{Y}{\consP{ID'}{N'}{T'}{\cnilp}} \otimes \remove{L'}{Y}{L''} ~\otimes  \\
                & \nateq{N}{N'} \otimes \tm{T} \\
     \multimap  & \{\ordQ{Q} \otimes \actP{A'}{L''} \otimes \tm{\cte{s}(T)} \} \\
\\
\cte{market/2:} & \ordQ{\cte{front}(\ordN{\cte{market}}{A}{P}{ID}{N}{T},Q)} \otimes \dual{A}{A'} \otimes \actP{A'}{L'} ~\otimes \\
                & \mktex{A}{L'}{Y} \otimes \prcQ{A'}{Y}{\consP{ID'}{N'}{T'}{(\consP{ID1}{N1}{T1}{L}})} ~\otimes \\
                & \nateq{N}{N'} \otimes \tm{T} \\
     \multimap  & \{ \ordQ{Q} \otimes \prcQ{A'}{Y}{(\consP{ID1}{N1}{T1}{L})} \otimes  \actP{A'}{L'} \otimes \tm{\cte{s}(T)} \} \\
\\
\cte{market/3:} & \ordQ{\cte{front}(\ordN{\cte{market}}{A}{P}{ID}{N}{T},Q)} \otimes \dual{A}{A'} \otimes \actP{A'}{L'} \otimes\phantom{a} \\
                & \mktex{A}{L'}{Y} \otimes \prcQ{A'}{Y}{\consP{ID'}{N'}{T'}{\cte{nilP}}} \otimes \remove{L'}{Y}{L''} ~\otimes\\
                &\natgreat{N}{N'} \otimes \natminus{N,N'}{N''} \\
     \multimap  & \{  \ordQ{\cte{front}(\ordN{\cte{limit}}{A}{P}{ID}{N''}{T},Q)} \otimes \actP{A'}{L''} \} \\
     \\
\cte{market/4:} & \ordQ{\cte{front}(\ordN{\cte{market}}{A}{P}{ID}{N}{T},Q)} \otimes \dual{A}{A'} \otimes \actP{A'}{L'} \otimes\phantom{a} \\
                & \mktex{A}{L'}{Y} \otimes \prcQ{A'}{Y}{\consP{ID'}{N'}{T'}{(\consP{ID1}{N1}{T1}{L}})} ~\otimes\\
                &\natgreat{N}{N'} \otimes \natminus{N,N'}{N''} \\
     \multimap  & \{ \ordQ{\cte{front}(\ordN{\cte{limit}}{A}{P}{ID}{N''}{T},Q)} ~\otimes \\
                & \prcQ{A'}{Y}{(\consP{ID1}{N1}{T1}{L})} \otimes \actP{A'}{L'} \} \\
     \\
\cte{market/5:} & \ordQ{\cte{front}(\ordN{\cte{market}}{A}{P}{ID}{N}{T},Q)} \otimes \dual{A}{A'} \otimes \actP{A'}{L'} ~\otimes\\
                & \mktex{A}{L'}{Y} \otimes \prcQ{A'}{Y}{\consP{ID'}{N'}{T'}{L}} ~\otimes\\
                &\natless{N}{N'} \otimes \natminus{N',N}{N''}  \otimes \tm{T} \\
     \multimap  & \{ \ordQ{Q} \otimes \prcQ{A'}{Y}{(\consP{ID'}{N''}{T'}{L})} \otimes \actP{A'}{L'} \otimes \tm{\cte{s}(T)} \}
\end{array}
\]
}
\caption{Exchanging market orders}
\label{fig:marketOrdersEx}
\end{figure*}

The rules for exchanging market orders are given in
Figure~\ref{fig:marketOrdersEx}.
The first rule, \cte{market/empty}, addresses the situation when a market order
is in the order queue, but there are no opposite resident orders to be matched
against. The order is then removed from the order queue.
Rules \cte{market/1} and \cte{market/2} address the case when the incoming
market order's quantity is the same as the quantity of the best available
opposite resident order. In this case these orders are simply exchanged, and
again we distinguish cases whether the resident order was the last in the
queue (\cte{market/1}) or not (\cte{market/2}).
Rules \cte{market/3} and \cte{market/4} address the case when the incoming
market order's quantity is greater than the quantity of the best available
opposite resident order, i.e., when $N>N'$. In this case the resident order is
exchanged and the rest of the market order remains in the order queue. The
two rules distinguish whether the resident order was the last in the price
queue (\cte{market/3}) or not (\cte{market/4}).
Finally, rule \cte{market/5} describes the situation when an incoming market
order is strictly less (quantity-wise) compared to the best opposite resident
order. The considered market order is exchanged completely whereas the resident
order only partially.

\subsection{Cancel Orders}
\label{subs:cancel}

Cancel order is an instruction to remove a particular resident order from the
trading system.  Cancel orders refer to a resident order by its identifier.  If,
by chance, the order to be canceled is not in the market, nothing happens and
the cancel order is removed from the order queue.  If it is there, it is removed
from the price queue. Similarly as in exchanging limit orders this results in two
sub-cases: the order is the last one in its price queue or not.

The rules for performing order canceling are given in Figure~\ref{fig:cancelOrders}.
%
\begin{figure*}[ht]
{\footnotesize
\[
\begin{array}{rl}
\cte{cancel/inListNil:} & \ordQ{\cte{front}(\ordN{\cte{cancel}}{A}{P}{ID}{N}{T},Q)} \otimes \actP{A}{L'} ~\otimes \\
                & \prcQ{A}{P}{\consP{ID}{N}{T'}{\cte{nilP}}} \otimes \remove{L'}{P}{L''} \otimes \tm{T} \\
     \multimap  & \{\ordQ{Q} \otimes \actP{A}{L''} \otimes \tm{\cte{s}(T)} \} \\
     \\
\cte{cancel/inListCons:} & \ordQ{\cte{front}(\ordN{\cte{cancel}}{A}{P}{ID}{N}{T},Q1)}~\otimes \\
                & \prcQ{A}{P}{Q} \otimes \inListF{Q}{ID} \otimes \removeF{Q}{ID}{Q'} \otimes \tm{T} \\
     \multimap  & \{\ordQ{Q1} \otimes \prcQ{A}{P}{Q'} \otimes \tm{\cte{s}(T)} \} \\
     \\
\cte{cancel/notInListQueue:} & \ordQ{\cte{front}(\ordN{\cte{cancel}}{A}{P}{ID}{N}{T},Q1)}~\otimes \\
                & \prcQ{A}{P}{Q} \otimes \notInListF{Q}{ID} \otimes \tm{T} \\
     \multimap  & \{\ordQ{Q1} \otimes \prcQ{A}{P}{Q} \otimes \tm{\cte{s}(T)} \} \\
     \\
\cte{cancel/notInListActive:} & \ordQ{\cte{front}(\ordN{\cte{cancel}}{A}{P}{ID}{N}{T},Q1)} \otimes \actP{A}{L'} ~\otimes \\
                & \notInList{L'}{P} \otimes \tm{T} \\
     \multimap  & \{\ordQ{Q1}  \otimes \actP{A}{L'} \otimes \tm{\cte{s}(T)} \}
\end{array}
\]
}
\caption{Cancel orders}
\label{fig:cancelOrders}
\end{figure*}

\section{Towards a Mechanized Verification of ATS Properties}
\label{sec:proof}

Using our formalization we are able to check that this combination of
order-matching rules does not violate some of the expected ATS
properties.
Although CLF is a powerful logical framework fit for specifying the
syntax and semantics of concurrent systems, stating and proving
properties about these systems goes beyond its current expressive
power. For this task, one needs to consider states of computation, and
the execution traces that lead from one state to another. Recent
developments show that CLF contexts (the states of computation) can be described
in CLF itself through the notion of generative grammars~\cite{Simmons12}.
These are grammars whose language consists of all possible CLF contexts which
satisfy the property being considered. The general idea is to show that every
reachable state consists of a context in this language.

Using such grammars plus reasoning on steps and traces of computation, it is
possible to state and prove meta-theorems about CLF specifications. This method
is structured enough to become the meta-reasoning engine behind
CLF~\cite{CervesatoS13}, and therefore it is used for the proofs in this paper.

\subsection{No Locked or Crossed Market}

Here we show that the $\cbid$ price $B$ is always less than the $\cask$ price
$S$ in any reachable state using the rules presented in the previous sections.
In other words, we show the following invariant:

\begin{property}[No locked-or-crossed market]
\label{prop.1}
If $\actP{\cbuy}{L_B}$ and $\actP{\csell}{L_S}$ and $\maxP{L_B}{B}$ and
$\minP{L_S}{S}$, then $B < S$.
\end{property}

Definition~\ref{grammarNLC} shows a grammar that generates contexts, or states,
satisfying Property~\ref{prop.1}. This is achieved by the guards
$\maxP{L_B}{B}$, $\minP{L_S}{S}$, $B < S$ on the rewriting rule \cte{gen/0} for
the start symbol $\cgenp{Q}{L_B}{L_S}{T}$.

\begin{defn}
\label{grammarNLC}
The following generative grammar $\signcl$\footnote{$NLC$ stands for
no locked-or-crossed.} produces only contexts where $\cbid < \cask$.

{\footnotesize
\[
\begin{array}{lcl}
\cte{gen/0}     &:& \cgenp{Q}{L_B}{L_S}{T} \otimes \maxP{L_B}{B} \otimes \minP{L_S}{S} \otimes B < S \\
                & & \multimap \{\ordQ{Q} \otimes \actP{\cbuy}{L_B} \otimes \actP{\csell}{L_S} \otimes \tm{T} \\
                & & \qquad \otimes~ \cte{gen-buy}(L_B) \otimes \cte{gen-sell}(L_S)\}.\\[1.5mm]
\cte{gen/buy1}  &:& \cte{gen-buy}(\cnilp) \multimap \{1\}.\\[1.5mm]
\cte{gen/buy2}  &:& \cte{gen-buy}(P::L_B) \multimap \{\prcQ{\cbuy}{P}{L} \otimes \cte{gen-buy}(L_B)\}.\\[1.5mm]
\cte{gen/sell1} &:& \cte{gen-sell}(\cnilp) \multimap \{1\}.\\[1.5mm]
\cte{gen/sell2} &:& \cte{gen-sell}(P::L_S) \multimap \{\prcQ{\csell}{P}{L} \otimes \cte{gen-sell}(L_S)\}.\\[1.5mm]
\end{array}
\]
}
\end{defn}

Intuitively, to show that the market is never in a locked-or-crossed state, we
show that, given a context generated by the grammar in
Definition~\ref{grammarNLC}, the application of an ATS rule (one step) will
result in another context that can also be generated by this grammar.
Coupled with the fact that computation starts at a valid context, this shows
that the property is always preserved. More formally, we will show the theorem:

\begin{thm}
\label{thm.1}
For every $\Delta \in L(\signcl)$ and rule $\sigma$, if 
$\Delta \xrightarrow{\sigma} \Delta'$, then $\Delta' \in L(\signcl)$.
\end{thm}

This theorem can be represented visually as:
\[
\xymatrix{
\cgenp{Q}{L_B}{L_S}{T} \ar[d]^{\epsilon} & \cgenp{Q'}{L_B'}{L_S'}{T'} \ar[d]^{\epsilon'}\\
\Delta \ar[r]^{\sigma}                   & \Delta'
}
\]
The proof consists in showing the existence of $\epsilon'$.

\begin{proof}
The proof proceeds by case analysis on $\sigma$.
We consider only rules that change the linear facts
$\actP{\cbuy}{L_B}$ and $\actP{\csell}{L_S}$, since otherwise we can
simply take $\epsilon' = \epsilon$ (possibly with different
instantiations for the variables $L$ in $\cte{gen/sell2}$ and
$\cte{gen/buy2}$).  Moreover, we restrict ourselves to the case of
incoming $\cbuy$ orders. The case for $\csell$ is analogous.

\paragraph{\fbox{Case $\sigma = \cte{limit/empty}$}}
This rule rewrites $L_B$, the list of buy prices, to a list $L_B'$
which extends $L_B$ by a new price $P$. Since ${\color{\bccolor}\cte{store}}$
was provable, we know that $P$ is less than the minimum sell price in the
market.
For \cte{limit/empty} to be applicable, we need:

{\small
\[
\Delta = \{ \ordQ{Q}, \actP{\cbuy}{L_B}, \actP{\csell}{L_S}, \tm{T} \} \cup \Delta_1
\]
}
\noindent
In which case we conclude that:

{\small
\[
\begin{array}{rrl}
Q    &= &\cte{front}(\ordN{\cte{\cbuy}}{A}{P}{ID}{N}{T},Q') \\
\epsilon &= &\cte{gen/0}(Q_g, L_B, L_S, T); \epsilon_1; \epsilon_2
\end{array}
\]}%
where $\epsilon_1; \epsilon_2$ rewrite $\cte{gen-buy}(L_B)$ and
$\cte{gen-sell}(L_S)$ to the $\prcQ{\_}{\_}{\_}$ facts that form $\Delta_1$.

\noindent
After applying the rule, the context is modified to:

{\small
\[
\begin{array}{rcl}
    \Delta' &=& \{ \ordQ{Q'}, \actP{\cbuy}{L_B'}, \actP{\csell}{L_S}, \tm{\cte{s}(T)}\}\\
    && \cup~ \{\prcQ{\cbuy}{P}{\consP{ID}{N}{T}{\cte{nilP}}}, \Delta_1\}
\end{array}
\]}%
where $L_B'$ is computed by the $\ins{L_B}{P}{L_B'}$ rule and consists of $L_B$
augmented by $P$.

\noindent
A derivation of $\Delta'$ can be obtained via the following steps:

{\small
\[
\epsilon' = \cte{gen/0}(Q, L_B', L_S, \cte{s}(T)); \epsilon_1; \cte{gen/sell2}; \epsilon_2
\]}%
where one extra step $\cte{gen/sell2}$, with $L = \consP{ID}{N}{T}{\cte{nilP}}$, 
generates $\prcQ{\cbuy}{P}{L}$. 
Observe that the guard $B < S$ in $\cte{gen/0}$ still
holds: if $\maxP{L_B'}{B}$ and $B \neq P$, then $B < S$ was part of the
assumption.  In case $\maxP{L_B'}{P}$, observe that $\store{\cbuy}{L_S}{P}$ only
holds if $P < S$, where $\minP{L_S}{S}$. This property is related only to
backward chaining predicates and can be proved in the LF framework Twelf using
standard techniques.

\paragraph{\fbox{Case $\sigma = \cte{limit/1}$}}
This rule rewrites $L_S$, the list of sell prices, to a list $L_S'$ which
consists of $L_S$ without a price $X$.
For $\cte{limit/1}$ to be applicable, we need:

{\small
\[
\begin{array}{rcl}
    \Delta &=& \{\ordQ{Q}, \actP{\cbuy}{L_B}, \actP{\csell}{L_S}, \tm{T}\}\\
        && \cup~ \{\prcQ{\csell}{X}{\consP{ID'}{N}{T'}{\cte{nilP}}}, ~\Delta_1\}
\end{array}
\]}%
where $X \in L_S$ is computed by $\exchange{\cbuy}{L_S}{P}{X}$.

\noindent
This can be derived by:

{\small
\[
\begin{array}{rrl}
Q      &= &\cte{front}(\ordN{\cte{limit}}{\cbuy}{P}{ID}{N}{T},Q')\\
\epsilon &= &\cte{gen/0}(Q_g, L_B, L_S, T); \epsilon_1; \cte{gen/sell2}; \epsilon_2
\end{array}
\]}%
where $\epsilon_1; \cte{gen/sell2}; \epsilon_2$ rewrite $\cte{gen-buy}(L_B)$ and
$\cte{gen-sell}(L_S)$ to the $\prcQ{\_}{\_}{\_}$ facts that form $\Delta_1$,
with the explicit $\cte{gen/sell2}$ generating the fact
$\prcQ{\csell}{X}{\consP{ID'}{N}{T'}{\cte{nilP}}}$.

\noindent
After applying the rule, the context is modified to:

{\small
\[
\Delta' = \{ \ordQ{Q'}, \actP{\cbuy}{L_B}, \actP{\csell}{L_S'}, \tm{\cte{s}(T)} \} \cup \Delta_1
\]}%
where $L_S'$ is computed by the $\remove{L_S}{X}{L_S'}$ rule and consists of
$L_S$ without $X$.

\noindent
A derivation of $\Delta'$ can be obtained via the following steps:

{\small
\[
\epsilon' = \cte{gen/0}(Q, L_B, L_S', \cte{s}(T)); \epsilon_1; \epsilon_2
\]}%
Since $L_S' \subset L_S$, then, considering $\minP{L_S}{S}$ and
$\minP{L_S'}{S'}$, it is the case that $S \leq S'$. Thus $B < S$ implies
$B < S'$.
This can be proved in Twelf given the specification of the appropriate relations
(such as $\subset$).

\paragraph{\fbox{Case $\sigma = \cte{limit/3}$}}
This case is analogous to $\cte{limit/1}$, except that the incoming order is
only partially exchanged because its quantity $N$ is greater than the quantity
$N'$ of the matching resident order.
The initial context is:

{\small
\[
\begin{array}{rcl}
\Delta &=& \{\ordQ{Q}, \actP{\cbuy}{L_B}, \actP{\csell}{L_S}, \tm{T}\}\\
        && \cup~ \{\prcQ{\csell}{X}{\consP{ID'}{N'}{T'}{\cte{nilP}}},\Delta_1\}
\end{array}
\]}%
where $X \in L_S$ is computed by $\exchange{\cbuy}{L_S}{P}{X}$.

\noindent
Which can be derived as before:

{\small
\[
\begin{array}{rrl}
Q     &= &\cte{front}(\ordN{\cte{limit}}{\cbuy}{P}{ID}{N}{T},Q''),~N>N' \\
\epsilon &= &\cte{gen/0}(Q, L_B, L_S, T); \epsilon_1; \cte{gen/sell2}; \epsilon_2
\end{array}
\]
}

\noindent
where $\epsilon_1; \cte{gen/sell2}; \epsilon_2$ rewrite $\cte{gen-buy}(L_B)$ and
$\cte{gen-sell}(L_S)$ to the $\prcQ{\_}{\_}{\_}$ facts that form $\Delta_1$,
with the explicit $\cte{gen/sell2}$ generating the fact
$\prcQ{\csell}{X}{\consP{ID'}{N'}{T'}{\cte{nilP}}}$.

\noindent
After applying the rule, the context is modified to:

{\small
\[
\Delta' = \{ \ordQ{Q'}, \actP{\cbuy}{L_B}, \actP{\csell}{L_S'}, \tm{\cte{s}(T)} \} \cup \Delta_1
\]
}
where $L_S'$ is computed by the $\remove{L_S}{X}{L_S'}$ rule and consists of
$L_S$ without $X$.

\noindent
A derivation of $\Delta'$ can be obtained via the following steps (note that
time does not change for this rule):

{\small
\[
\begin{array}{rcl}
    Q' &=& \cte{front}(\ordN{\cte{limit}}{\cbuy}{P}{ID}{N-N'}{T},Q'')\\
    \epsilon' &=& \cte{gen/0}(Q', L_B, L_S', T); \epsilon_1; \epsilon_2
\end{array}
\]}%
Since $L_S' \subset L_S$, then, considering $\minP{L_S}{S}$ and
$\minP{L_S'}{S'}$, it is the case that $S \leq S'$. Thus $B < S$ implies
$B < S'$.
As before, this argument can be developed in Twelf.\\

\noindent
The cases for market orders, $\sigma = \cte{market/\_}$, are analogous to the
$\cte{limit/\_}$ cases above. As for the case $\sigma = \cte{cancel/inListNil}$
- this rule rewrites $L_B$, the list of buy prices, to $L_B'$ which is $L_B$
without a price $P$. This case is analogous to $\cte{limit/1}$ and
$\cte{limit/3}$, except that $\epsilon$ contains an application of
$\cte{gen/buy2}$ which is deleted to obtain $\epsilon'$.

\end{proof}

\subsection{Exchanges happen at $\cbid$ or  $\cask$}

Every incoming order will either be exchanged or, if this is not possible,
stored as a resident order to be exchanged when it is matched (if ever).
The {\color{\bccolor}\cte{exchange}} predicate will be provable exactly when the
incoming order can be exchanged, i.e., there exists an opposite resident order at an
``acceptable'' price. The acceptable price is: lower than the price of an
incoming \cbuy{} order, or greater than the price of an incoming \csell{} order.
Since orders have an associated quantity, the exchange may partially or totally
consume the orders.

The price at which the exchange takes place is the best possible with respect to
the \emph{incoming} order.
Consequently, an incoming \cbuy{} order is exchanged at the minimal sell price
($\cask$), while an incoming \csell{} order is exchanged against the maximal
available buy price ($\cbid$).
Therefore whenever an order is exchanged it happens at either $\cbid$ or $\cask$
price, and \emph{only} at that price. In this section we show that this is
indeed the case for our encoding.  We will consider only the rules specifying
exchange of \emph{limit} orders (Figure~\ref{fig:limitOrdersEx}). The case for
\emph{market} orders follows in a similar fashion, but considering the predicate
{\color{\bccolor}mktExchange} instead of {\color{\bccolor}exchange}.

\begin{property}
\label{prop.2}
All exchanges happen at and only at the price $\cbid$ or $\cask$.
\end{property}

Note that Property~\ref{prop.1} is a property of the reachable \emph{states}, while
Property~\ref{prop.2} is concerns \emph{transitions} between states.
We can split it into two parts:
(1) in every exchange, only one resident order of price $X$ is consumed;
and
(2) $X$ is $\cbid$ or $\cask$.

\smallskip

Part (1) can be more formally stated as:

\begin{thm}
\label{thm:bidask.1}
Let $\Delta, \Delta' \in L(\signcl)$, and $\sigma$ be one of the exchange
rules $\cte{limit/i}$ for $1 \leq \mathtt{i} \leq 5$.
If $\Delta \xrightarrow{\sigma} \Delta'$, then for all
$\prcQ{A}{Y}{L} \in \Delta$ if $Y \neq X$ then $\prcQ{A}{Y}{L} \in \Delta'$,
where $X$ is determined by $\exchange{A'}{L}{P}{X}$ on the left side of rule
$\sigma$.
\end{thm}

\begin{proof}
By inspection of the rules $\cte{limit/i}$, we observe that the only facts of
the shape $\prcQ{A}{Z}{L}$ involved in the rewriting are those where $Z = X$,
where $X$ is bound by $\exchange{A'}{L}{P}{X}$.
\end{proof}

Part (2) of the property can be stated more precisely as:

\begin{thm}
\label{thm:bidask.2}
If $\exchange{\cbuy}{L}{P}{X}$ then $\minP{L}{X}$.\\
If $\exchange{\csell}{L}{P}{X}$ then $\maxP{L}{X}$.
\end{thm}

The statement concerns only a backward chaining predicate, so the proof follows
standard meta-reasoning techniques from LF, and can be implemented in a few
lines in Twelf.

\smallskip

Taken together, Theorems~\ref{thm:bidask.1} and~\ref{thm:bidask.2} guarantee
Property~\ref{prop.2}.

\section{Conclusion and Future Work}
\label{sec:conclusion}

We have formalized the core rules guiding the trade on exchanges
worldwide. We have done this by formalizing an archetypal automated
trading system in the concurrent logical framework CLF, with an
implementation in Celf.

Encoding orders in a market as linear resources results in
straightforward rules that either consume such orders when they are
bought/sold, or store them in the market as resident orders.  Moreover
the specification is modular and easy to extend with new order types,
which is often required in practice. This was our experience when
adding market and immediate-or-cancel types of orders to the
system. The concurrent aspect of CLF simulates the non-determinism
when orders are accumulated in the order queue, but, as explained,
orders from the queue are processed sequentially\footnote{As far as we
  know, no real life trading system performs parallel order matching
  and execution.}.

Using our formalization we were able to prove two standard properties
about a market working under these rules. First we proved that at any
given state the $\cbid$ price is smaller than the $\cask$, i.e., the
market is never in a locked-or-crossed state. Secondly we showed that the
trade always take place at $\cbid$ or $\cask$. The first property was proved
using generative grammars, an approach motivated by our goal to automate
meta reasoning on CLF specifications (not implemented in the current
version of Celf).  Recent investigations indicate that this approach
can handle many meta-theorems~\cite{Simmons12,CervesatoS13} related to
\emph{state} invariants, and ours is yet another example. The second
property is a combination of: (1) a property of a backward chaining
predicate; and (2) a \emph{transition} invariant. The former can be
proved using established methods in LF (in fact, we have proved the
desired property in Twelf). The second can be verified by inspection 
of the rules: the only linear facts that change in the next state, are
those rewritten on the right side of $\multimap$.

This specification is an important case study for developing the
necessary machinery for automated reasoning in CLF\@. It is one more
evidence of the importance of quantification over steps and traces of
a (forward-chaining) computation. It is interesting to note that our
example combines forward and backward-chaining predicates, but the
generative grammar approach still behaves well. In part because we are
only concerned with a linear part of the context. Nevertheless, the proof
still relied on some properties of backward chaining predicates.
In the second proof, this is even more evident. This indicates, unsurprisingly,
that meta-reasoning of CLF specifications must include the already
developed meta-reasoning of LF specifications.
In the meantime, we are investigating other properties of financial
systems that present interesting challenges for meta-reasoning, such
as showing that the price/time priority is respected upon exchange.

The difference between the proofs presented might be an indication that we need
to follow a more general approach than the one used in LF\@. In that framework,
most theorems that motivated the work have the same shape and their proofs
follow the same strategy. Therefore, it is possible to save the user a
lot of work by asking them to specify only the necessary parts that fills in a
``proof template'', which is then checked mechanically. When working with more
ad-hoc systems, such as the case of financial exchanges, the properties and
proofs are less regular, making it harder to figure out a good ``template'' that
fits all properties of interest. In this case, it may be beneficial to leave
more freedom (and consequently more work) to the user, in order to allow more
flexible meta-reasoning.


Concurrently, we plan to formalize other models of financial trading
systems, as this is a relevant application addressing some critical challenges.

\bibliographystyle{eptcs}
\bibliography{references}

\end{document}